\documentclass[12pt]{article}
\usepackage{amssymb,amsthm,bm,stmaryrd,leqno}

\theoremstyle{plain}
\newtheorem{theo}{Theorem}

\newtheorem{cor}[theo]{Corollary}

\newtheorem{lem}[theo]{Lemma}

\theoremstyle{definition}

\newtheorem{defi}[theo]{Definition}

\begin{document}

\title{Kolmogorov-Sinai entropy from the ordinal viewpoint}

\author{Karsten Keller, Mathieu Sinn\\
Institute of Mathematics, University of L\"ubeck}

{\large
Kolmogorov-Sinai entropy from the ordinal viewpoint} \\
\\
Karsten Keller (corresponding author) and Mathieu Sinn \\
Institute of Mathematics \\ University of L\"ubeck \\ Wallstra{\ss}e 40\\
23560 L\"ubeck,
Germany\\ \\
E-mail: keller@math.uni-luebeck.de

\newpage

\maketitle

\begin{abstract}
\noindent In the case of ergodicity much of the structure of a one-dimensional time-discrete dynamical system is already determined by its ordinal
structure. We generally discuss this phenomenon by considering the distribution of ordinal patterns, which describe the up and down in the orbits of
a Borel measurable map on a subset of the real numbers. In particular, we give a natural ordinal description of Kolmogorov-Sinai entropy of a large
class of one-dimensional dynamical systems and relate Kolmogorov-Sinai entropy to the permutation entropy recently introduced by Bandt and Pompe.
\end{abstract}


\section{Introduction}\label{intro}
The quantification of complexity is an important topic in the theory and application of dynamical systems. Various complexity measures have been
introduced and studied, like the correlation dimension given by Grassberger and Procaccia \cite{grass83} and the Kolmogorov-Sinai entropy (see
\cite{Walters82}). Both are well-motivated quantities, but they are not easy to estimate from real data. Recently, Bandt and Pompe
\cite{BandtPompe02} have proposed to measure complexity of one-dimensional dynamical systems on the base of the distribution of ordinal patterns in
the system.

Given a one-dimensional time-discrete dynamical system, defined by a map on a subset of the real line, ordinal patterns of order $d$ classify,
roughly speaking, points according to the order type of the vectors consisting of the points and of their first $d$ iterates. The higher $d$ is, the
more information is contained in the patterns. When the dynamical system is equipped with a probability measure which is invariant with respect to
the given map, this measure yields a distribution on the finite set of ordinal patterns of some order. The permutation entropy introduced by Bandt
and Pompe \cite{BandtPompe02} is the Shannon entropy of the ordinal pattern distribution of order $d$ relative to $d$ for $d\to\infty$ and provides a
relatively simple concept for measuring complexity in a `standardized' way. It is remarkable that in the case of a piecewise monotone interval map
the permutation entropy coincides with the Kolmogorov-Sinai entropy (see Bandt et al.~\cite{bkp}). Also note that a similar result has been given for
the topological entropy (see Bandt et al.~\cite{bkp} and, for some further results, Misiurewicz \cite{Misiu03}).

From the viewpoint of estimating complexity, ergodic dynamical systems are most important because according to the Birkhoff Ergodic Theorem many
properties of such systems can be reconstructed from a single orbit with probability one. For interval maps and their multidimensional
generalizations, Amig\'{o} et al.~\cite{Amigo05} show equality of a modification of the permutation entropy to the Kolmogorov-Sinai entropy. Their
concept of permutation entropy uses a non-standardized sequence of finer and finer partitions of the considered interval as the Kolmogorov-Sinai
entropy itself.

In the following we discuss consequences of invariancy and ergodicity in a one-dimensional dynamical system for the ordinal structure of the system.
On this base, we give an ordinal description of the Kolmogorov-Sinai entropy for ergodic systems and continuous systems on compact sets and
investigate the relationship of permutation entropy and Kolmo\-go\-rov-Sinai entropy.

\paragraph{Ordinal patterns.} The central concept of this paper is the concept of an ordinal pattern.
We give it here in a slightly altered form which is more convenient for the following considerations.

\begin{defi}\label{odef}
For $d\in {\mathbb N}=\{1,2,3,\ldots \}$, we call ${\cal I}_d:=\vartimes_{l=1}^d\{0,1,\ldots ,l\}$ the {\em set of ordinal patterns of order} $d$.
Further, for $X\subset {\mathbb R}$ and a map $f: X\hookleftarrow$ the sequence ${\bf i}_d={\bf i}_d(x)=(i_{d,1}(x),i_{d,2}(x),\ldots ,i_{d,d}(x))\in
{\cal I}_d$ defined by
\begin{eqnarray}\label{defform}
i_{d,l}(x)=\#\{r\in \{0,1,\ldots ,l-1\}\,|\,f^{\circ (d-r)}(x)\leq f^{\circ(d-l)}(x)\}
\end{eqnarray}
for $l=1,2,\ldots,d$ is called the {\em ordinal pattern of order} $d\in {\mathbb N}$ {\em realized by} $x\in X$. Here $f^{\circ n}(x)$ for
$n\in{\mathbb N}$ denotes the $n$-th iterate of a point $x\in X$ with respect to $f$ where $f^{\circ 0}(x)=x$.
\end{defi}

The vector ${\bf i}_d(x)$ of inversion numbers just describes the order relations between the components of the vector $(x, f(x), f^{\circ 2}(x),
\ldots, f^{\circ d}(x))$. Note that originally ordinal patterns were given in terms of permutations which express the order relations more directly
than the inversion numbers (see \cite{BandtPompe02}). However, the latter representation is better from the computational viewpoint (see \cite{kes}),
and the main point is that both descriptions of ordinal patterns provide exactly the same partitions of the interested system. In particular, both
representations contain the information on whether $f^{\circ k}(x)\leq f^{\circ l}(x)$ for $k,l\in\{0,1,\ldots,d\}$ with $l<k$.

\paragraph{Partitions and entropies.} Let $X\subset {\mathbb R}$, let $f: X\hookleftarrow$
be a ${\mathbb B}(X)$-${\mathbb B}(X)$-measurable map, where ${\mathbb B}(X)$ denotes the Borel $\sigma$-algebra on $X$, and let $\mu$ be an {\em
$f$-invariant} probability measure on ${\mathbb B}(X)$, that is
\begin{eqnarray*}
\mu(f^{-1}(A)) = \mu(A)
\end{eqnarray*}
for each $A\in{\mathbb B}(X)$. For the following we refer to the simple fact that finite $\sigma$-algebras on $X$ are in one-to-one relation to
finite partitions of $X$, since there is exactly one finite partition generating a finite $\sigma$-algebra.

Recall the definition of the {\em (Shannon) entropy} of a finite partition ${\cal A}=\{A_1,A_2,\linebreak\ldots ,A_k\}\subset {\mathbb B}(X)$ of $X$
and of the from ${\cal A}$ generated $\sigma$-algebra ${\cal B}$, respectively, given by
\begin{eqnarray*}
H({\cal B})=H({\cal A})= - \sum_{j=1}^k \mu (A_j)\ln\mu (A_j).
\end{eqnarray*}
For a finite partition ${\cal A}\subset {\mathbb B}(X)$ of $X$, let $h_\mu (f,{\cal A}):=\lim_{n\to\infty} \frac{1}{n}H(\bigvee_{j=0}^{n-1} f^{-\circ
j}({\cal A}))$, where $\bigvee_{j=0}^{n-1} f^{-\circ j}({\cal A})$ denotes the $\sigma$-algebra generated by the set of all $f^{-\circ j}(A)$ with
$j\in \{0,1,\ldots ,n-1\}$ and $A\in {\cal A}$. It is well known that this limit exists (see, e.g.,~\cite{Walters82}, Cor.~4.9.1). The {\em
Kolmogorov-Sinai entropy} of $f$ is defined by
\begin{eqnarray*}
h_\mu (f):=\sup\{h_\mu (f,{\cal A})\,|\,{\cal A}\subset {\mathbb B}(X)\mbox{ finite partition of }X\}.
\end{eqnarray*}

The Kolmogorov-Sinai entropy is an important theoretical concept in ergodic theory, however, its estimation from an orbit of a dynamical system, is
complicated since principally arbitrarily fine partitions have to be considered. The permutation entropy introduced by Bandt and Pompe
\cite{BandtPompe02} seems to be an interesting alternative.

\begin{defi}
Let $X\subset {\mathbb R}$, let $f: X\hookleftarrow$ be a ${\mathbb B}(X)$-${\mathbb B}(X)$-measurable map, and let $\mu$ be an $f$-invariant
probability measure on ${\mathbb B}(X)$. For $d\in {\mathbb N}$, let
\begin{eqnarray*}
{\cal P}_d=\{P_{\bf i}\,|\,{\bf i}\in {\cal I}_d\}\mbox{ with }P_{\bf i}=\{x\in X\,|\,{\bf i}_d(x)={\bf i}\}.
\end{eqnarray*}
The {\em permutation entropy} of $f$ is defined by
\begin{eqnarray*}
h^\ast_\mu (f)=\limsup_{d\to\infty}\frac{1}{d}\,H({\cal P}_d).
\end{eqnarray*}
\end{defi}
Note that for each $d\in {\mathbb N}$ and ${\bf i}\in {\cal I}_d$, the set $P_{\bf i}$ indeed belongs to ${\mathbb B}(X)$. This immediately follows
from the facts that $\{x\in X\,|\,f^{\circ (d-r)}(x)\leq f^{\circ(d-l)}(x)\}= (f^{\circ(d-l)}-f^{\circ (d-r)})^{-1}(X\cap[0,\infty[)$, the mapping
$f^{\circ(d-l)}-f^{\circ (d-r)}$ is ${\mathbb B}(X)$-${\mathbb B}(X)$-measurable and ${X\cap[0,\infty[}\in{\mathbb B}(X)$ (compare formula
(\ref{defform})).

\section{Generating properties}

The key for proving the main results of the paper is that in the ergodic case ordinal patterns separate points, in other words, the partitions
defined by ordinal pattern are `generating'. Note that finite generating partitions play an important role in ergodic theory (see \cite{Walters82},
\S 4.6). By considering ordinal patterns of increasing order we obtain a sequence of finite partitions ${\cal P}_1,{\cal P}_2,{\cal P}_3,\ldots$
generating increasing $\sigma$-algebras, that is $\sigma({\cal P}_1)\subset\sigma({\cal P}_2)\subset\sigma({\cal P}_3)\subset\ldots$.

The following lemma shows how ordinal information of a dynamical system with invariant measure can be extracted from its orbits. Recall that if
$X\subset {\mathbb R}$ and $f: X\hookleftarrow$ is ${\mathbb B}(X)$-${\mathbb B}(X)$-measurable, then an $f$-invariant probability measure $\mu$ on
${\mathbb B}(X)$ is called {\em ergodic} if
\begin{eqnarray*}
\mu(A)\in\{0,1\}\mbox{ for each } A\in{\mathbb B}(X)\mbox{ with }f^{-1}(A)=A.
\end{eqnarray*}
It is well-known that this is equivalent to the statement that $\mu(f^{-1}(B)\,\Delta\, B)=0$ for $B\in{\mathbb B}(X)$ implies $\mu(B)\in\{0,1\}$,
where $\Delta$ denotes the symmetric difference of sets (see \cite{Walters82}, Th.~1.5.).

\begin{lem}\label{mainlem}
Let $X\subset {\mathbb R}$, let $f: X\hookleftarrow$ be a ${\mathbb B}(X)$-${\mathbb B}(X)$-measurable map, and let $\mu$ be an $f$-invariant
probability measure on ${\mathbb B}(X)$.

Then there exists a set $\widetilde{X}\in {\mathbb B}(X)$ with $\mu (\widetilde{X})=1$ and $f(\widetilde{X})\subset\widetilde{X}$, such that
\begin{eqnarray*}
\bm{\alpha}(x):=\lim_{d\to\infty}\frac{i_{d,d}(x)}{d}
\end{eqnarray*}
is well-defined for all $x\in \widetilde{X}$. If $\mu$ is ergodic, $\widetilde{X}$ can be chosen such that moreover $\bm{\alpha}(x)=\mu(\{y\in
\widetilde{X}\,|\,y\leq x\})$
and $\bm{\alpha}$ is injective on $\widetilde{X}$.
\end{lem}
\begin{proof}
Let $z\in X$, and for $x\in X$ let
\begin{eqnarray*}
\bm{\alpha}(z,x):= \lim_{d\to\infty} \frac{\#\{i\in\{1,2,\ldots ,d\}\,|\,f^{\circ i}(x)\leq z\}}{d}
\end{eqnarray*}
if the limit exists. Fix a set $X_z\in {\mathbb B}(X)$ with $\mu (X_z)=1$ consisting of points $x\in X$ for which $\bm{\alpha}(z,x)$ exists.
Moreover, in case of ergodicity of $f$ require that $\mu (\{y\in X\,|\,y\leq z\})=\bm{\alpha}(z,x)$ for all $x\in X_z$. Both is possible by the
Birkhoff Ergodic Theorem (see, e.g.,~\cite{Walters82}, Th.~1.14.~and text below). For $\widetilde{X}_z:=\bigcap_{n\in {\mathbb N}}f^{-\circ n}(X_z)$
it holds $f(\widetilde{X}_z)\subset\widetilde{X}_z$ and, since $\mu$ is $f$-invariant, $\mu (\widetilde{X}_z)=1$. Moreover, since $\bm{\alpha}(z,x)$
is defined if and only if $\bm{\alpha}(z,f(x))$ is, $\bm{\alpha}(z,x)$ is defined for all $x\in\widetilde{X}_z$.

Let now $Z$ be a countable dense subset of $X$ containing all points of $X$ isolated from the left or from the right, that is, every $z\in X$ with
$]z-\epsilon,z[\ \cap\ X=\emptyset$ or $]z,z+\epsilon[\ \cap\ X=\emptyset$ for some $\epsilon>0$. Let $\widetilde{X}=\bigcap_{z\in
Z}\widetilde{X}_z$. Obviously, $\mu(\widetilde{X})=1$, and $f(\widetilde{X})\subset\widetilde{X}$. We show that $\bm{\alpha}(x,x)$ is defined for
$x\in\widetilde{X}$. If $x\in \widetilde{X} \cap Z$, this is obvious. For the case where $x\in \widetilde{X}\setminus Z$, consider
\begin{eqnarray*}
{\cal E} = \{\,]-\infty,z] \,|\, z\in Z \}
\end{eqnarray*}
and let $\sigma({\cal E})$ denote the $\sigma$-algebra on ${\mathbb R}$ generated by ${\cal E}$. Obviously,
\begin{eqnarray*}
P_x(B) := \lim_{d\to\infty} \frac{\#\{i\in\{1,2,\ldots,d\}\,|\,f^{\circ i}(x) \in B\}}{d}
\end{eqnarray*}
is well-defined for $B\in{\cal E}$ and can uniquely be extended to a probability measure on the measurable space $({\mathbb R}, \sigma({\cal E}))$,
which we still denote by $P_x$. Now let $(y_k)_{k\in {\mathbb N}}$ be a monotonically increasing sequence in $Z$ and $(z_k)_{k\in {\mathbb N}}$ be a
monotonically decreasing sequence in $Z$ such that $\lim_{k\to\infty}y_k=\lim_{k\to\infty}z_k=x$. (By definition of $Z$ such sequences always exist.)
By definition of $\widetilde{X}$, the numbers $\bm{\alpha}(y_k,x), \bm{\alpha}(z_k,x) $ are defined for all $k\in {\mathbb N}$. Obviously, the
sequences $( \bm{\alpha}(y_k,x) )_{k\in {\mathbb N} }$ and $( \bm{\alpha}(z_k,x) )_{k\in {\mathbb N} }$ are monotonically nonincreasing and
monotonically nondecreasing, respectively, therefore the limits $ \lim_{k\to\infty}\bm{\alpha}(y_k,x)$ and $\lim_{k\to\infty}\bm{\alpha}(z_k,x) $
both exist. Further, $P_x(\{x\}) = \lim_{d\to\infty} \frac{\#\{i\in\{1,2,\ldots,d\}\,|\,f^{\circ i}(x) = x\}}{d}$ is obviously well-defined. In
particular, $P_x(\{x\}) = \frac{1}{m}$ if $m$ is the least integer such that $f^{\circ m}(x) = x$ and $P_x(\{x\}) = 0$ if no such integer exists.
Hence
\begin{eqnarray*}
\bm{\alpha}(y_k,x) + P_x(\{x\}) &\leq&
\liminf_{d\to\infty} \frac{\#\{i\in\{1,2,\ldots,d\}\,|\,f^{\circ i}(x)\leq x\}}{d} \\
&\leq& \limsup_{d\to\infty} \frac{\#\{i\in\{1,2,\ldots,d\}\,|\,f^{\circ i}(x)\leq x\}}{d} \ \leq \ \bm{\alpha}(z_k,x)
\end{eqnarray*}
for every $k\in{\mathbb N}$. By the upper continuity of finite measures one obtains
\begin{eqnarray*}
\lim_{k\to\infty}\Big(\bm{\alpha}(z_k,x) - (\bm{\alpha}(y_k,x) + P_x(\{x\}))\Big) = \lim_{k\to\infty} P_x(]y_k,z_k]) - P_x(\{x\}) = 0.
\end{eqnarray*}
Consequently, $\lim_{d\to\infty} \frac{\#\{i\in\{1,2,\ldots,d\}\,|\,f^{\circ i}(x)\leq x\}}{d}$ exists, and hence $\bm{\alpha}(x) = \bm{\alpha}(x,x)
$ is defined for each $x\in\widetilde{X}$ and coincides with $\lim_{k\to\infty}\mu (\{y\in X\,|\,y\leq z_k\})=\mu (\{y\in X\,|\,y\leq x\})$ for $\mu$
ergodic.

In particular, $\bm{\alpha}$ is monotonically nondecreasing in the ergodic case. Therefore, if $J=\bm{\alpha}^{-1}(\alpha)$ consists of more than one
point for some $\alpha\in {\mathbb R}$, then $J$ is the intersection of $\widetilde{X}$ with some interval, hence there are at most countable many of
such $\alpha$. In the case $\mu (J)>0$, necessarily, $\min J$ exists and $\mu (J)=\mu(\{\min J\})$. Otherwise, one would find some $x,y\in J$ with
$x<y$ such that $\mu(\,]-\infty,x]\,\cap\,J) < \mu(\,]-\infty,y]\,\cap\,J)$, leading to
\begin{eqnarray*}
\bm{\alpha}(x)=\mu(\,]-\infty,x]\,\cap\,X)<\mu(\,]-\infty,y]\,\cap\,X) =\bm{\alpha}(y)
\end{eqnarray*}
in contradiction to $\bm{\alpha}$ being constant on $J$. Let
\begin{eqnarray*}
U&=&\bigcup\,\{J\,|\,J=\bm{\alpha}^{-1}(\alpha)\mbox{ for some }\alpha\in {\mathbb R},\,\#J>1,\mu
(J)=0\}\\
&\quad&\cup\ \bigcup\,\{J\setminus\{\min J\}\,|\,J=\bm{\alpha}^{-1}(\alpha)\mbox{ for some }\alpha\in {\mathbb R},\,\#J>1,\mu (J)>0\}
\end{eqnarray*}
Clearly, $\mu (U)=0$. By omitting all points with iterates in $U$ the set $\widetilde{X}$ can be modified such that $\bm{\alpha}$ is injective.
\end{proof}

\noindent {\it Remarks:}\vspace{1mm}\\
1. The proof of the lemma is given under relatively general assumptions. If $\mu$ is ergodic, then ${\bm \alpha}$ is just the distribution function
of $\mu$, which can easily seen by the `ergodic part' of the Birkhoff Ergodic Theorem (see, e.g.,~the remark below Th.~1.14.~in \cite{Walters82}) and
the Clivenko-Cantelli argument (see, e.g.,~\cite{Pollard84}, sec.~II.2.). It remains to show injectivity of ${\bm \alpha}$ on a set of full measure,
which is the last part of our proof. The case that $X$ is compact and $f$ is continuous can reduced to the ergodic case argueing with ergodic
decomposition (see
e.g.~\cite{Walters82}, \S 6.2.)\vspace{1mm}\\
2. Bandt and Shiha \cite{BandtShiha07} have used the above statement for ergodic $\mu$ with con\-tinuous distribution function, in order to show that
all finite-dimensional distributions of a stationary stochastic process can be reconstructed from its ordinal structure and the one-dimensional
distribution.\vspace{3mm}

As the above lemma shows, ergodic case ordinal patterns are `generating' in the ergodic case:

\begin{cor}
Let $X\subset {\mathbb R}$, let $f: X\hookleftarrow$ be a ${\mathbb B}(X)$-${\mathbb B}(X)$-measurable map, and let $\mu$ be an $f$-invariant ergodic
probability measure on ${\mathbb B}(X)$.

Then there exists a set $\widetilde{X}\subset X$ with $\mu (\widetilde{X})=1$ and $f(\widetilde{X})\subset \widetilde{X}$, such that different points
of $\widetilde{X}$ have different ordinal patterns of some order $d$.
\end{cor}

We will need the supplement to Lemma \ref{mainlem} below, where $\bigvee_{d\in {\mathbb N}} {\cal P}_d$ denotes the $\sigma$-algebra on $X$ generated
by $\bigcup_{d\in {\mathbb N}} {\cal P}_d$.

\begin{lem}\label{suppl}
Let $\widetilde{X}$ be defined as in Lemma \ref{mainlem}, ${\bm \alpha}$ be considered as a map on $\widetilde{X}$, and $\widetilde{\cal P}$ denote
the restriction of $\bigvee_{d\in {\mathbb N}} {\cal P}_d$ to $\widetilde{X}$. Then ${\bm \alpha}$ is $\widetilde{\cal P}$-${\mathbb
B}([0,1])$-measurable.
\end{lem}
\begin{proof}
For $\alpha\,\in\,[0,1]$ and $\widetilde{X}_\alpha:=\{x\in \widetilde{X}\,|\, \bm{\alpha}(x) \leq \alpha \}$ it holds
\begin{eqnarray*}
\widetilde{X}_\alpha &=& \{x\in \widetilde{X}\,|\,\mbox{For all }q\in {\mathbb Q}\,\cap\, ]x_0,\infty [\mbox{ there exists a } d_0\in {\mathbb
N}\\& &\hspace{5cm}\mbox{ with }\frac{i_{d,d}(x)}{d}<q\mbox{ for all }d\geq d_0\}\\
&=& \left ( \bigcap_{q\,\in\, {\mathbb Q}\,\cap\, ]x_0,\infty [}\ \bigcup_{d_0\in {\mathbb N}}\ \bigcap_{d\,\geq\, d_0} \{x\in
X\,|\,i_{d,d}(x)<d\hspace{0.15mm}q\}\right )\cap\widetilde{X}.
\end{eqnarray*}
Since for each $d\in {\mathbb N}, q\in{\mathbb Q}\,$, the set $\{x\in X\,|\,i_{d,d}(x)<d\hspace{0.15mm}q\}$ belongs to the $\sigma$-algebra generated
by ${\cal P}_d$, the set $\bigcap_{q\,\in\, {\mathbb Q}\,\cap\, ]x_0,\infty [}\bigcup_{d_0\in {\mathbb N}} \bigcap_{d\,\geq\, d_0}\{x\in
X\,|\,i_{d,d}(x)<d\hspace{0.15mm}q\}$ is contained in $\bigvee_{d\in {\mathbb N}} {\cal P}_d$. The system $\{]0,\alpha],|\,\alpha\in [0,1]\}$
generates ${\mathbb B}([0,1])$, and hence the result follows.
\end{proof}

\section{Kolmogorov-Sinai entropy and order structure}

As well known, for a map $f$ and an invariant measure $\mu$ the Kolmogorov-Sinai entropy can be obtained as $\lim_{n\to\infty}h_\mu (f,{\cal A}_n)$
if ${\mathbb B}(X)$ is generated by a sequence of finite partitions ${\cal A}_n$ with corresponding increasing $\sigma$-algebras (see
\cite{Walters82}, Th.~4.22). The following theorem says that under relative mild assumptions this sequence can be chosen in a standard way, namely by
considering the partitions of $X$ with respect to ordinal patterns of order $d$ for increasing $d$.

\begin{theo}\label{important}
Let $X\subset {\mathbb R}$, let $f: X\hookleftarrow$ be a ${\mathbb B}(X)$-${\mathbb B}(X)$-measurable map, and let $\mu$ be an $f$-invariant
probability measure on ${\mathbb B}(X)$. If $\mu$ is ergodic or $X$ is compact and $f$ continuous, then
\begin{eqnarray}
h_\mu (f)=\lim_{d\to\infty} h_\mu (f,{\cal P}_d)=\sup_{d\in {\mathbb N}} h_\mu (f,{\cal P}_d).\label{hmu}
\end{eqnarray}
\end{theo}

\begin{proof}
We start with the case that $\mu$ is ergodic. Clearly, $\bigvee_{d\in {\mathbb N}} {\cal P}_d$ is contained in ${\mathbb B}(X)$ (see end of Section
\ref{intro}). Let $\widetilde{X}$ be chosen as for the ergodic case in Lemma \ref{mainlem} and let ${\bm \alpha}$ be considered as a map on
$\widetilde{X}$. Then ${\bm \alpha}$ is injective, monotone and, by Lemma \ref{suppl}, $\widetilde{\cal P}$-${\mathbb B}([0,1])$-measurable.
Therefore, one easily sees that $\widetilde{\cal P}$ coincides with the restriction of ${\mathbb B}(X)$ to $\widetilde{X}$. Hence for each $B\in
{\mathbb B}(X)$ there exists an $A\in \bigvee_{d\in {\mathbb N}} {\cal P}_d$ with $\mu (A\,\Delta\,B)=0$, and (\ref{hmu}) follows from Theorem 4.22
in \cite{Walters82}.

In the case that $\mu$ is non-ergodic use ergodic decomposition (see, e.g.,~\cite{Walters82}, \S 6.2.) and the Monotone Convergence Theorem (see
e.g.~\cite{Walters82}, Theorem 0.8): $\mu$ can be written as $\mu=\int_Em\, d\tau (m)$, where $E$ denotes the set of ergodic $f$-invariant Borel
probability measures and $\tau$ is a Borel measure on the $f$-invariant Borel probability measures on $X$ (with respect to the weak$^\ast$ topology)
with $\tau (E)=1$. According to results shown by Jacobs (compare \cite{Walters82}, Theorem 8.4), it holds
\begin{eqnarray}
\hspace{-1cm}h_{\mu}(f,{\cal A})=\int_E h_m(f,{\cal A})\, d\tau (m)\mbox{ for each finite partition }{\cal A}\subset {\mathbb B}(X)\mbox{ of }X,\\
\hspace{-1cm}h_{\mu}(f)=\int_E h_m(f)\, d\tau (m).
\end{eqnarray}
By the already shown, the sequence $(g_d)_{d\in {\mathbb N}}$ of functions $g_d$ on the $f$-invariant Borel probability measures defined by
$g_d(m)=h_m(f,{\cal P}_d)$ on $E$ and vanishing outside of $E$ is monotonically increasing and converges pointwisely to $h_m(f)$ on $E$. By the
Monotone Convergence Theorem, either $(h_{\mu}(f,{\cal P}_d))_{d\in {\mathbb N}}$ is bounded and it holds (\ref{hmu}), or $(h_{\mu}(f,{\cal
P}_d))_{d\in {\mathbb N}}$ is unbounded and $h_m(f)$ is infinite on a set of positive measure with respect to $\tau$. In the latter case
$h_{\mu}(f)=\infty$.\end{proof}

\section{Permutation entropy}

As a relatively simple consequence of the above theorem, we obtain that the Kolmogo\-rov-Sinai entropy is not larger than the permutation entropy.
This generalizes a result given by Amig\'{o} et al.~\cite{Amigo05} for ergodic maps $f$ on intervals with finite $h_\mu (f)$.

\begin{cor}\label{ineq}
Let $X\subset {\mathbb R}$, let $f: X\hookleftarrow$ be a ${\mathbb B}(X)$-${\mathbb B}(X)$-measurable map, and let $\mu$ be an $f$-invariant
probability measure on ${\mathbb B}(X)$. If $\mu$ is ergodic or $X$ is compact and $f$ is continuous, then $h_\mu (f)\leq h^\ast_\mu (f)$.
\end{cor}

\begin{proof}
If $h_\mu (f)=0$, we are done. Otherwise, consider some $a>0$ with $h_\mu (f)>a$. Given some $c>1$ with $h_\mu (f)>c\,a$, by the above theorem there
exists a $d_0\in {\mathbb N}$ with the following property: For each $d\geq d_0$ there exist some $n(d)$ with $ \frac{1}{n}H(\bigvee_{j=0}^{n-1}
f^{-\circ j}({\cal P}_d))>c\,a$ for all $n\geq n(d)$.

Let $d\geq d_0$ and $n\geq \max \{n(d),\frac{d}{c-1}\}$. Note that the $\sigma$-algebra generated by ${\cal P}_{d+n}$ contains\linebreak
$\bigvee_{j=0}^{n-1} f^{-\circ j}({\cal P}_d)$. This can be seen from the fact that for each ${\bf i}=(i_1,i_2,\ldots,i_d)\in {\cal I}_d$ and
$j\in\{0,1,\ldots ,n-1\}$ one has
\begin{eqnarray*}
f^{-\circ j}(P_{\bf i})=\bigcap_{l=1}^d \{x\in X\,|\, i_{d,\,l}(f^{\circ j}(x))=i_l\}=\bigcap_{l=1}^d \{x\in X\,|\, i_{d+j,\,l}(x)=i_l\}.
\end{eqnarray*}
Therefore, one easily sees that $f^{-\circ j}(P_{\bf i})$ is contained in the $\sigma$-algebra generated by sets $\{x\in X\,|\,f^{\circ k}(x)\leq
f^{\circ l}(x)\}$ with $k,l\in\{0,1,\ldots,d+n\}$ and $l<k$ (compare to (\ref{defform})). Since validity of $f^{\circ k}(x)\leq f^{\circ l}(x)$ only
depends on the ordinal pattern of $x$ (see below definition of ordinal patterns), such sets are the union of sets in ${\cal P}_{d+n}$, hence the set
$f^{-\circ j}(P_{\bf i})$ belongs to the $\sigma$-algebra generated by ${\cal P}_{d+n}$.

Consequently, $H({\cal P}_{d+n})\geq H(\bigvee_{j=0}^{n-1} f^{-\circ j}({\cal P}_d))$, which implies
\begin{eqnarray*}
\frac{1}{d+n}\,H({\cal P}_{d+n})&\geq&\frac{1}{d+n}\,H(\bigvee_{j=0}^{n-1} f^{-\circ j}({\cal P}_d))\\&\geq&
\frac{1}{(c-1)\,n+n}\,H(\bigvee_{j=0}^{n-1} f^{-\circ
j}({\cal P}_d))\\
&=& \frac{1}{cn}\,H(\bigvee_{j=0}^{n-1} f^{-\circ
j}({\cal P}_d))\\
&>&a.
\end{eqnarray*}
From this one obtains $h^\ast_\mu (f)\geq a$ and, since $a$ can be chosen arbitrarily near to $h_\mu (f)$, it follows $h^\ast_\mu (f)\geq h_\mu (f)$.
\end{proof}
The central question under which assumptions Kolmogorov-Sinai entropy and permutation entropy coincide remains open in the general case, whereas for
piecewise monotone interval maps with invariant probability measure coincidence has been shown by Bandt et al.~\cite{bkp}, as already mentioned in
the Introduction. A map $f$ from an interval $[a,b]$ is said to be piecewise monotone if there are points $c_0,c_1,\ldots , c_k$ with
$a=c_0<c_1<\ldots <c_k=b$ such that $f$ is continuous and strictly monotone on $[c_{j-1},c_j]$ for each $j=1,2,\ldots ,k$. (Continuity of the whole
map is not required, in contrast to the often given definition of a piecewise monotone interval map.) Note that the proofs in \cite{bkp} do not
provide an idea for a generalization since they essentially base on piecewise monotonicity, in particular, on a result given by Misiurewicz and
Szlenk \cite{Misiurewicz80}.

As mentioned in the Introduction, Amig\'{o} et al.~\cite{Amigo05} have given a modified concept of permutation entropy of an interval map, for which
they have shown equality to the Kolmogorov-Sinai entropy. Their idea is to consider permutation entropy for a stationary stochastic process with
values in a finite set equipped with an arbitrary order. (The resulting entropy does not depend on the order.) Given a partition of the interval on
which the map is defined into finitely many subintervals, one only regards in which of the subintervals a point lies. This yields the desired
stochastic process, and the permutation entropy of the particular partition is naturally defined. The permutation entropy of the interval map is then
defined as the limit of permutation entropies of partitions with interval lengths tending to $0$.

Just like the Kolmogorov-Sinai entropy, this modified concept of permutation entropy uses infinitely many partitions which in contrast to the
original definition of permutation entropy are non-standardized. Note that the result of Amig\'{o} et al.~\cite{Amigo05} is given not only for
interval maps but in a multidimensional context.


\begin{thebibliography}{99}

\bibitem{Amigo05}
J.~M.~Amig\'{o}, M.~B.~Kennel, L.~Kocarev, The permutation entropy rate equals the metric entropy rate for ergodic information sources and ergodic
dynamical systems, {\it Physica D} 210 (2005), 77--95.

\bibitem{BandtShiha07}
C.~Bandt, F.~Shiha, Order patterns in time series, {\it J. Time Ser. Anal.} 28 (2007), 646-665.

\bibitem{BandtPompe02}
C.~Bandt, B.~Pompe, Permutation entropy: A natural complexity measure for time series, {\it Phys. Rev. Lett.} 88 (2002), 174102.

\bibitem{grass83}
P. Grassberger \and I. Procaccia, Measuring the strangeness of strange attractors, {\it Physica D} 9 (1983), 189--208.

\bibitem{bkp}
C. Bandt, G. Keller \and B.~Pompe, Entropy of interval maps via permutations, {\it Nonlinearity} 15 (2002), 1595--1602.

\bibitem{kes}
K. Keller, J. Emonds \and M.~Sinn, Time series from the ordinal viewpoint, {\it Stochastics and Dynamics} 2 (2007), 247--272.

\bibitem{Misiu03}
M. Misiurewicz, Permutations and topological entropy for interval maps, {\it Nonlinearity} 16 (2003), 971--976.

\bibitem{Misiurewicz80}
M. Misiurewicz \and W. Szlenk, Entropy of piecewise monotone mappings, {\it Stud. Math.} 67 (1980), 45--63.

\bibitem{Pollard84}
D.~Pollard, {\it Convergence of Stochastic Processes}, Springer, Berlin, 1984.

\bibitem{Walters82}
P. Walters, {\it An Introduction to Ergodic Theory}, Springer, New York, 1982.

\end{thebibliography}
\end{document}